\documentclass[pra,twocolumn,amsmath,amssymb,nofootinbib]{revtex4}

\usepackage[bookmarks = false, pdfpagemode = None, pdfstartview = FitH, colorlinks = true, urlcolor = blue,hyperfootnotes=false]{hyperref}
\usepackage{amsthm}
\usepackage{mathrsfs}
\usepackage{upgreek}
\usepackage{color}

\providecommand{\abs}[1]{\left\lvert#1\right\rvert}

\newcommand{\tr}{\textup{\textrm{tr}}}
\newcommand{\set}[1]{\mathcal{#1}}
\newcommand{\code}{\texttt}

\newcommand{\bra}[1]{{\left\langle{#1}\right\vert}}
\newcommand{\ket}[1]{{\left\vert{#1}\right\rangle}}

\newcommand{\brakket}[3]{\left\langle#1\left|\vphantom{#1}#2\vphantom{#3}\right|#3\right\rangle}

\newcommand{\bke}[1]{\textbf{#1}}
\renewcommand{\bke}[1]{}

\newcommand{\proj}[1]{{\mathbb{#1}}}
\newcommand{\opfunc}[1]{{\mathrm{#1}}}
\newcommand{\func}[1]{{\mathrm{#1}}}
\newcommand{\supop}[1]{{\mathrm{#1}}}
\newcommand{\prop}[1]{{\mathscr{#1}}}
\newcommand{\bigominus}{\mathop{\raisebox{-.1em}{\Large\boldmath$\ominus$}}}

\newtheorem{lem}{Lemma}
\newtheorem{defn}{Definition}
\newtheorem{thm}{Theorem}

\newcommand{\notG}{{\mathrel{\text{$\set{G}$\llap{$\backslash$}}}}}

\definecolor{gray}{gray}{0.5}

\begin{document}

\title{Simulating Concordant Computations}

\author{Bryan Eastin}
\email{beastin@nist.gov}

\affiliation{National Institute of Standards and Technology,
Boulder, CO 80305}

\begin{abstract}
A quantum state is called concordant if it has zero quantum discord with respect to any part.  By extension, a concordant computation is one such that the state of the computer, at each time step, is concordant.
In this paper, I describe a classical algorithm that, given a product state as input, permits the efficient simulation of any concordant quantum computation having a conventional form and composed of gates acting on two or fewer qubits.
This shows that such a quantum computation must generate quantum discord if it is to efficiently solve a problem that requires super-polynomial time classically.
While I employ the restriction to two-qubit gates sparingly, a crucial component of the simulation algorithm appears not to be extensible to gates acting on higher-dimensional systems.
\end{abstract}

\maketitle

The search for the origin of the computational power of quantum mechanics has proven to be a recurring theme in quantum information theory.  Primarily, this search has focused on identifying the feature of quantum mechanics that permits the efficient\footnote{The definition of ``efficient'' is taken from classical computer science, where it refers to any computation that requires an amount of resources (particularly time steps) scaling at most polynomially with the problem size.} solution of certain classically intractable problems.  In addition to being useful, computational speedups of this magnitude are intriguing since, classically, no such improvement is to be found over rather basic models of computation, e.g., the Turing machine.

Among the proposed sources of this quantum advantage, the most widely studied is a kind of non-local correlation known as entanglement~\cite{Horodecki09}.  The state of a composite system is entangled if it cannot be described in terms of a, possibly uncertain, local assignment of states to individual subsystems.  Classically, non-trivial correlations indicate imperfect information about the state of the system, but entanglement is possible for quantum states of maximal knowledge, or pure states.  At the extreme, an entangled state of a composite system may be pure while the marginal state of the component subsystems is maximally impure, or maximally mixed.  In other words, one may know everything possible about the state of a composite quantum system without knowing anything about the state of the component subsystems.  As a distinctly non-classical property and a necessary resource for protocols such as teleportation and quantum error correction, entanglement is a natural suspect when investigating the power of quantum computing.

There are two kinds of evidence in favor of entanglement as the crucial resource for achieving speedups that enable the efficient solution of a classically intractable problem, a variety of speedup henceforth labeled Promethean.
First, there are proofs that pure-state quantum computations generating only limited amounts of entanglement can be efficiently simulated classically and are therefore incapable of solving any problem that cannot be solved in polynomial time by a classical computer.  An early result of this sort was shown by Jozsa and Linden~\cite{Jozsa03}, who described a method for efficiently simulating any quantum computation whose correlations are approximately confined to regions of bounded size.  Shortly thereafter, Vidal proposed an efficient simulation algorithm for quantum computations whose maximum Schmidt rank for any bipartition of the computer scales at most as a polynomial~\cite{Vidal03}.  These methods of simulation can each be applied to quantum computations with either mixed or pure states, but in the former case classical correlations, in addition to entanglement, are restricted.  The second kind of evidence for the importance of entanglement is its apparent generation by all implementations of Shor's quantum factoring algorithm.  In particular, a typical implementation of Shor's algorithm has been shown to generate entanglement that precludes its simulation by either Jozsa and Linden's or Vidal's method~\cite{Jozsa03,Orus04}.  To summarize, entanglement is necessary for obtaining Promethean speedups with pure-state quantum computing, and there are indications that it may be required for Shor's algorithm.

Regarding mixed states, further, and contrary, evidence comes from the DQC1 model of quantum computation~\cite{Knill98b}, where all but one of the qubits in the computer is initially prepared in the maximally mixed state.
DQC1 is believed to be strictly less powerful than pure-state quantum computing~\cite{Knill98b,Ambainis06}, but it nonetheless seems to be capable of providing Promethean speedups in, for example, trace estimation.  Datta, Flammia, and Caves have shown numerically that trace estimation is possible even with a vanishing amount of entanglement (as measured by the negativity of bipartite splittings)~\cite{Datta05}.  Nevertheless, Datta and Vidal have shown that the Schmidt rank grows exponentially for certain bipartitions of a quantum computer performing trace estimation~\cite{Datta07}, thereby demonstrating the existence of correlations, though not necessarily entanglement, sufficient to thwart Vidal's simulation method.  Based on these results, it seems probable that Promethean speedups are possible even in the absence of entanglement.

But if entanglement is not the source of Promethean speedups in DQC1 then we are left to ask what is.  Among the proposed alternatives is a measure of non-classical correlation known as quantum discord~\cite{Zurek00}.  Datta, Shaji, and Caves have shown that discord is indeed present in the trace-estimation algorithm~\cite{Datta08}, but it has never been proven to be necessary.
The work presented in this paper was motivated by the desire to show that discord is necessary for Promethean speedups in mixed-state quantum computations.  Since, for pure states, discord reduces to a measure of entanglement, this would amount to an extension of the result (described above) about the utility of entanglement in pure-state quantum computing.   To this end, I considered the difficulty of simulating concordant computations, i.e., those that generate no quantum discord, as suggested by Ref.~\cite{Lanyon08}.

Here, I describe an algorithm for efficiently simulating, using a classical computer, any computation that does not generate discord and consists of a sequence of one- and two-qubit unitary gates followed by single-qubit measurements.  Section~\ref{sec:notation} briefly introduces some notation and Sec.~\ref{sec:concordance} covers discord, concordance, and concordant computations and proves a few results that are employed later.  My simulation algorithm is described for quantum computations in a conventional form in Sec.~\ref{sec:simulatingConventionalConcordantComputations} and extensions to non-conventional forms are discussed in Sec.~\ref{sec:extensions}.  The conclusion contains a discussion of open problems.

\section{Notation\label{sec:notation}}

Unitary operators, projectors, and sets are denoted by capital roman letters in math-italic, black-board, and calligraphic font, respectively, e.g., $U$, $\proj{P}$, and $\set{A}$.  For more generic functions on quantum states I use capital Roman letters in math font.
Throughout the paper, quantum operators and states are given subscripts (which may be sets) to denote the subsystems they act upon and/or to index the component corresponding to that subsystem; all other identifying indices and labels are represented as superscripts.
Thus, the state of a composite system can be expressed as $\rho_{\set{A}\set{B}}$, where $\set{A}$ and $\set{B}$ are disjoint sets indexing the subsystems, and the marginal density operator of part $\set{B}$ of $\rho_{\set{A}\set{B}}$ is written as $\rho_\set{B} = \tr_\set{A}(\rho_{\set{A}\set{B}})$, where $\tr_\set{A}$ is the trace over part $\set{A}$.  Contrary to this example, I frequently omit the subscript when it would specify the entire system.  Whenever indicated, the time step is labeled by a superscript.
The symbols $\cup$, $\cap$, $\setminus$, and $\ominus$ are used to denote the set-theoretic operations of union, intersection, difference, and symmetric difference, and I denote the complement of a set $\set{G}$ by $\notG$.
Vectors over finite fields are denoted by placing a right arrow over a symbol, and the subscripting of such vectors by a set represents the restriction of the vector to the components indicated by the set, e.g., $\vec{i}_\set{G} = \{i_k: k\in\set{G}\}$.
The support of an operator is taken to mean the set of subsystems upon which the operator acts nontrivially.

\section{Concordance\label{sec:concordance}}

The notion of a classical state frequently carries with it the idea of a preferred basis.  In a Stern-Gerlach experiment, for example, the resulting superposition of different spins and locations is rarely considered as simply representing a novel basis for classical particles.
From this perspective, a classical state is one selected from a preferred basis of orthogonal states, where the basis for a composite system arises from the tensor product of the preferred bases for the component subsystems.  When the state of a system is uncertain, we describe it using a probability distribution over known, or pure, classical states.

A concordant state differs from this definition of classicality only in that no preferred basis is specified; any set of orthogonal bases for the subsystems may be used to determine the pure states allowed to the composite system.  I take a concordant computation, in turn, to be one in which the state of the computer after any step is concordant.  This usage of ``concordant'' seems to have been coined by Andrew White, but it has not previously appeared in publication.  In the following subsections, I explicitly define concordant states and computations as well as reviewing or proving some results used later in the paper.

\subsection{Quantum discord\label{subsec:quantumDiscord}}

Quantum discord is a measure of non-classical correlations introduced by Zurek~\cite{Zurek00}.  Intuitively, it quantifies the amount of non-local disturbance caused by measuring part of a quantum state.
For a quantum state $\rho_{\set{A}\set{B}}$, the quantum discord with respect to part $\set{B}$ can be defined as
\begin{align*}
  \begin{split}
    \opfunc{D}_\set{B}(\rho_{\set{A}\set{B}}) = \min_{\{\proj{P}^i_\set{B}\}} &\left[\opfunc{H}\!\!\left(\rho_{\set{A}\set{B}}^{\{\proj{P}^i_\set{B}\}}\right)-\opfunc{H}\!\!\left(\rho_{\set{B}}^{\{\proj{P}^i_\set{B}\}}\right)\right]\\
    &- \left[\opfunc{H}(\rho_{\set{A}\set{B}})-\opfunc{H}(\rho_\set{B})\right]
  \end{split}
\end{align*}
where $\{\proj{P}^i_\set{B}\}$ is a complete set of orthogonal one-dimensional projectors (CSOOP) on part $\set{B}$,
\begin{align*}
  \rho_{\set{A}\set{B}}^{\{\proj{P}^i_\set{B}\}} &= \sum_i \proj{P}^i_\set{B}\rho_{\set{A}\set{B}}\proj{P}^i_\set{B}\;,
\end{align*}
and $\opfunc{H}(\rho) = -\tr(\rho\log_2\rho)$ is the Von Neumann entropy, the quantum analog of Shannon entropy.  This definition is somewhat less general than that of Zurek, who did not insist on the minimization, instead making quantum discord a function of the choice of projectors.

Ollivier and Zurek~\cite{Ollivier01} showed that $\opfunc{D}_\set{B}(\rho_{\set{A}\set{B}})=0$ if and only if
\begin{align}
  \rho_{\set{A}\set{B}} = \sum_i \proj{P}^i_\set{B} \rho_{\set{A}\set{B}} \proj{P}^i_\set{B} \label{eq:commutesWithProjector}
\end{align}
for some CSOOP $\{\proj{P}^i_\set{B}\}$ on part $\set{B}$, or equivalently,
\begin{align}
  \rho_{\set{A}\set{B}} = \sum_i \tr_\set{B}(\rho_{\set{A}\set{B}}\proj{P}^i_\set{B})\otimes \proj{P}^i_\set{B} = \sum_i p_i \rho_{\set{A}}^{\proj{P}^i_\set{B}}\otimes \proj{P}^i_\set{B}
\end{align}
where $p_i = \tr(\rho_{\set{A}\set{B}}\proj{P}^i_\set{B})$, $\rho_{\set{A}}^{\proj{P}^i_\set{B}} = \tr_\set{B}(\rho_{\set{A}\set{B}}^{\proj{P}^i_\set{B}})$, and
\begin{align}
  \rho_{\set{A}\set{B}}^{\proj{P}^i_\set{B}} &= \proj{P}^i_\set{B}\rho_{\set{A}\set{B}}\proj{P}^i_\set{B}/\tr(\rho_{\set{A}\set{B}}\proj{P}^i_\set{B})\;.
\end{align}
Lemma~\ref{lem:projectorsUniqueUpToPartialDegeneracy} shows that the set of projectors satisfying Eq.~\ref{eq:commutesWithProjector} is unique up to degeneracy in part $\set{B}$ of $\rho_{\set{A}\set{B}}$.  The notion of degeneracy on a part of a larger state is clarified by Definition~\ref{def:partDegeneracy}.

\begin{defn}
Two states are degenerate on part $\set{B}$ of $\rho_{\set{A}\set{B}}$ if the corresponding projectors $\proj{P}_\set{B}$ and $\proj{Q}_\set{B}$ satisfy $\tr_\set{B}(\rho_{\set{A}\set{B}}\proj{P}_\set{B}) = \tr_\set{B}(\rho_{\set{A}\set{B}}\proj{Q}_\set{B})$.
\label{def:partDegeneracy}
\end{defn}

\begin{lem}
Given two CSOOPs on $\set{B}$, $\{\proj{P}^i_\set{B}\}$ and $\{\proj{Q}^j_\set{B}\}$, and a state $\rho_{\set{A}\set{B}} = \sum_i \proj{P}^i_\set{B} \rho_{\set{A}\set{B}} \proj{P}^i_\set{B}$, $\rho_{\set{A}\set{B}} = \sum_j \proj{Q}^j_\set{B} \rho_{\set{A}\set{B}} \proj{Q}^j_\set{B}$ if and only if $\tr_\set{B}(\rho_{\set{A}\set{B}}\proj{P}^i_\set{B}) = \tr_\set{B}(\rho_{\set{A}\set{B}}\proj{Q}^j_\set{B})$ for all $\proj{P}^i_\set{B}\proj{Q}^j_\set{B}\neq 0$.
\label{lem:projectorsUniqueUpToPartialDegeneracy}
\end{lem}
\begin{proof}
The forward implication follows from
\begin{align*}
  \begin{split}
    \tr_\set{B}(\rho_{\set{A}\set{B}}\proj{P}^i_\set{B}\proj{Q}^j_\set{B}) &=
    \sum_h \tr_\set{B}(\proj{P}^h_\set{B} \rho_{\set{A}\set{B}} \proj{P}^h_\set{B}\proj{P}^i_\set{B}\proj{Q}^j_\set{B}) \\
    &= \tr_\set{B}(\rho_{\set{A}\set{B}}\proj{P}^i_\set{B}\proj{Q}^j_\set{B}\proj{P}^i_\set{B}) = e_{ij} \tr_\set{B}(\rho_{\set{A}\set{B}}\proj{P}^i_\set{B})\\
    &= \sum_h \tr_\set{B}(\proj{Q}^h_\set{B} \rho_{\set{A}\set{B}} \proj{Q}^h_\set{B}\proj{P}^i_\set{B}\proj{Q}^j_\set{B}) \\
    &= \tr_\set{B}(\rho_{\set{A}\set{B}}\proj{Q}^j_\set{B}\proj{P}^i_\set{B}\proj{Q}^j_\set{B}) = e_{ij} \tr_\set{B}(\rho_{\set{A}\set{B}}\proj{Q}^j_\set{B})
  \end{split}
\end{align*}
where $e_{ij}=\tr_\set{B}(\proj{P}^i_\set{B}\proj{Q}^j_\set{B})$.
The reverse implication follows from
\begin{align*}
  \begin{split}
    \rho_{\set{A}\set{B}} &= \sum_i \tr_\set{B}(\rho_{\set{A}\set{B}}\proj{P}^i_\set{B})\otimes \proj{P}^i_\set{B}\\
      &= \sum_{i,j} \tr_\set{B}(\rho_{\set{A}\set{B}}\proj{P}^i_\set{B})\otimes(\proj{P}^i_\set{B} \proj{Q}^j_\set{B})\\
      &= \sum_{i,j} \tr_\set{B}(\rho_{\set{A}\set{B}}\proj{Q}^j_\set{B})\otimes(\proj{P}^i_\set{B} \proj{Q}^j_\set{B}) \\
      &= \sum_{j} \tr_\set{B}(\rho_{\set{A}\set{B}}\proj{Q}^j_\set{B})\otimes \proj{Q}^j_\set{B}\;.
  \end{split}
\end{align*}
\end{proof}

If the quantum discord of $\rho_{\set{A}\set{B}}$ is zero with respect to both $\set{A}$ and $\set{B}$ then, by two applications of Eq.~\ref{eq:commutesWithProjector},
\begin{align}
  \rho_{\set{A}\set{B}} = \sum_{i,j} p_{i j} \proj{P}^i_\set{A}\otimes \proj{P}^j_\set{B} \label{eq:commutesWithProjectors}
\end{align}
for some CSOOPs $\{\proj{P}^i_\set{A}\}$ and $\{\proj{P}^j_\set{B}\}$.  For fixed $\{\proj{P}^i_\set{A}\}$, Lemma~\ref{lem:projectorsUniqueUpToPartialDegeneracy} shows that the set of projectors $\{\proj{P}^j_\set{B}\}$ satisfying Eq.~\ref{eq:commutesWithProjectors} is unique up to the degeneracy common to all $\rho_\set{B}^{\proj{P}^i_\set{A}}$, that is, up to degeneracy appearing in each of the subblocks of $\rho$ projected out by some $\proj{P}^i_\set{A}$.

\subsection{Concordant states}

The adjective ``concordant'' is intended to indicate a lack of quantum discord.  Because discord is an asymmetric, bipartite measure, however, it is not completely obvious what this restriction ought to mean with regard to quantum states, especially states of composite systems composed of more than two subsystems.  I choose to label a state as concordant if it has zero discord with respect to any part.  This is codified in the following definition.

\begin{defn}
A state $\rho$ is concordant if $\opfunc{D}_\set{A}(\rho)=0$ for any strict subset $\set{A}$ of the subsystems of $\rho$.\label{def:concordantState}
\end{defn}

In particular, Def.~\ref{def:concordantState} guarantees that $\opfunc{D}_{k}(\rho)=0$ for any $k$ labeling a single subsystem of some concordant state $\rho$.  By
Eq.~\ref{eq:commutesWithProjector}, this implies that, for any concordant state $\rho$, there exists a CSOOP $\{\proj{P}^{i}_k\}$ for every subsystem $k$ such that
\begin{align}
  \rho = \sum_{i} \proj{P}^{i}_k \rho \proj{P}^{i}_k\;. \label{eq:concordantStateCriterion}
\end{align}
An equivalent form of the implication that often proves useful is
\begin{align}
  \rho = \sum_{\vec{i}} \proj{P}^{\vec{i}} \rho \proj{P}^{\vec{i}} = \sum_{\vec{i}} p_{\vec{i}} \proj{P}^{\vec{i}}\label{eq:concordantStateCriterion2}
\end{align}
where $\proj{P}^{\vec{i}} = \prod_{k} \proj{P}^{i_k}_k$ and $\{\proj{P}^{i_k}_k\}$ for fixed $k$ is a CSOOP for the $k$th subsystem.

The reasoning above shows that Def.~\ref{def:concordantState} implies Eq.~\ref{eq:concordantStateCriterion2}, but conversely, any state satisfying Eq.~\ref{eq:concordantStateCriterion2} clearly satisfies Def.~\ref{def:concordantState}.  Thus, Eq.~\ref{eq:concordantStateCriterion2} can be taken as an alternate definition of a concordant state.  In words, a state is concordant if there exists a product basis, that is, a basis arising from the tensor product of local orthogonal bases, such that its density operator is diagonal.

\subsection{Concordant computations\label{subsec:concordantComputations}}

In keeping with standard practice, I adopt a description of quantum computation based on the quantum circuit model, where the evolution of the state of a system is described by a sequence of operators.
Most generally, the operations applied can be chosen probabilistically, based, for example, on the path the computation has taken thus far, as revealed by measurements.
In this model, it is natural to label a computation as concordant if the state of the computer is concordant both initially and after each step of the evolution, a notion formalized below.

\begin{defn}
A quantum computation described by a sequence of operators $\{\supop{G}^t\}$ acting on some input state $\rho^0$ is concordant if each state $\rho^t = \supop{G}^t \circ \cdots \circ\supop{G}^2 \circ \supop{G}^1 (\rho^0)$ is concordant for every path of the computation.\label{def:quantumComputation}
\end{defn}

Being concordant, each computational state might be considered classical for some choice of the classical basis, but a concordant computation is slightly more general than a randomized classical computation in that the product eigenbasis can change from one step to the next.

Definition~\ref{def:quantumComputation} is problematic for questions of computational complexity since it is possible to obscure the difficulty of an algorithm by employing very complex operations or initial states.  The specification of an arbitrary input state $\rho^0$, for example, entails a quantity of real numbers exponential in the number of subsystems, even if $\rho^0$ is concordant.  (See Ref.~\cite{Jozsa03} for a careful treatment of the difficulties posed by the use of real numbers.)  I avoid these problems and simplify the following discussion by initially considering only computations that are conventional, as defined by Def.~\ref{def:conventionalQuantumComputation}.  In Sec.~\ref{sec:extensions} I discuss ways in which the restriction to conventional computations can be relaxed.

\begin{defn}
A conventional quantum computation consists of an input product state diagonal in the standard basis, $\rho^0 = \bigotimes_k \rho^0_k$, followed by a sequence of unitary gates $\{G^t\}$, and concluded by single-subsystem measurements determining the outcome of the computation.  Each $\rho^0_k$ and $G^t$ (when restricted to its support) is required to be efficiently computable.\label{def:conventionalQuantumComputation}
\end{defn}

The evolution of a concordant computation of the form given by Def.~\ref{def:conventionalQuantumComputation} is particularly simple.
Because the spectrum of a density operator is invariant under conjugation by unitary operators, any unitary gate can be considered simply as a change of eigenbasis for the density operator.  For a concordant computation, there is guaranteed to exist a product basis, both before and after a gate, such that the density operator describing the state of the computer is diagonal.  Thus, the effect of any unitary operator can be, at most, to change the product eigenbasis and permute the associated eigenvalues.

More specifically, Lemma~\ref{lem:concordantComputationsAreSimple} shows that a transformation between concordant states induced by a unitary gate with support $\set{G}$ is equivalent to a change of product eigenbasis on $\set{G}$ together with a permutation with support $\set{G}$ of the vectors indexing the eigenvalues.  In general, the unitary gate will not actually be a permutation followed by a change of product eigenbasis but merely be equivalent to one for the given initial state.

\begin{lem}
If $\sigma = G \rho G^\dagger$ where $G$ is a unitary operator with support $\set{G}$, $\rho$ and $\sigma$ are concordant, and
$\rho = \sum_{\vec{i}} p_{\vec{i}} \proj{P}^{\vec{i}}$ then $\sigma = \sum_{\vec{j}} q_{\vec{j}} \proj{P}^{\vec{j}}_{\notG} \proj{Q}^{\vec{j}}_\set{G}$
where $q_{P\cdot\vec{i}} = p_{\vec{i}}$ for some permutation $P$ with support $\set{G}$.\label{lem:concordantComputationsAreSimple}
\end{lem}
\begin{proof}\ \\

Since $\sigma$ is concordant there exists $\{\proj{Q}^{\vec{j}}\}$ such that
\begin{align*}
  \sigma = \sum_{\vec{j}_\set{G}} \proj{Q}^{\vec{j}}_\set{G} \sigma \proj{Q}^{\vec{j}}_\set{G}\;,
\end{align*}
where $\proj{Q}^{\vec{j}}_\set{G} = \prod_{k\in\set{G}} \proj{Q}^{j_k}_k$ and likewise for subsequent similar projectors. Moreover,
\begin{align*}
  \begin{split}
    \sum_{\vec{i}_\notG} \proj{P}^{\vec{i}}_\notG \sigma \proj{P}^{\vec{i}}_\notG &= \sum_{\vec{i}_\notG} \proj{P}^{\vec{i}}_\notG G \rho G^\dagger \proj{P}^{\vec{i}}_\notG \\
    &= G \sum_{\vec{i}_\notG} \proj{P}^{\vec{i}}_\notG \rho \proj{P}^{\vec{i}}_\notG G^\dagger = G \rho G^\dagger = \sigma \;. \label{eq:unchangedBasesForUntouchedSubsystems}
  \end{split}
\end{align*}
Thus, $\sigma$ can be written in the form
\begin{align*}
  \sigma = \sum_{\vec{j}} \proj{P}^{\vec{j}}_{\notG} \proj{Q}^{\vec{j}}_\set{G} \sigma \proj{Q}^{\vec{j}}_\set{G} \proj{P}^{\vec{j}}_{\notG} = \sum_{\vec{j}} q_{\vec{j}} \proj{P}^{\vec{j}}_{\notG} \proj{Q}^{\vec{j}}_\set{G}\;.
\end{align*}

To see that the specified permutation exists, consider a graph $\Gamma$ where the nodes correspond to the projectors $\proj{Q}^{\vec{j}}_\set{G}$ and $G \proj{P}^{\vec{i}}_\set{G} G^\dagger$ and two nodes are connected if their associated projectors are not orthogonal.
Since $\{\proj{Q}^{\vec{j}}_\set{G}\}$ and $\{G \proj{P}^{\vec{i}}_\set{G} G^\dagger\}$ project onto two eigenbases for the state
\begin{align*}
  {\sigma}^{\proj{P}^{\vec{j}}_\notG}_\set{G} \propto \sum_{\vec{j}_\set{G}} q_{\vec{j}} \proj{Q}^{\vec{j}}_\set{G} = \sum_{\vec{j}_\set{G}} p_{\vec{j}} G \proj{P}^{\vec{j}}_\set{G} G^\dagger\;,
\end{align*}
projectors connected in $\Gamma$ are associated, by the uniqueness properties of the spectral decomposition, with the same eigenvalue of ${\sigma}^{\proj{P}^{\vec{j}}_\notG}_\set{G}$ and therefore with the same eigenvalues of $\sigma$.
Two spectral decompositions of the same density operator are related by a unitary transformation, so each connected component of $\Gamma$ includes an equal number of projectors from $\{\proj{Q}^{\vec{j}}_\set{G}\}$ and $\{G \proj{P}^{\vec{i}}_\set{G} G^\dagger\}$.  Thus, it is possible to assign $q_{\vec{j}} = p_{\vec{i}}$ where $\vec{j}=P\cdot\vec{i}$, $P$ is a permutation such that $\vec{j}_\notG=\vec{i}_\notG$, and $\{\proj{Q}^{\vec{j}}_\set{G}\}$ and $\{G \proj{P}^{\vec{i}}_\set{G} G^\dagger\}$ are in the same connected component of $\Gamma$.
\end{proof}

\begin{figure*}
\code{
\begin{tabbing}
01\hspace{2ex}\=For \=each\= \ subsystem $k$:\\
02   \>\>Choose $i_k$ according to the probability distribution $\func{Pr}[i_k=w] = \brakket{w}{{U^0_k}^\dagger \rho^0_k U^0_k}{w}$.\\
03   \>$\vec{j} := P \cdot \vec{i}$\\
04   \>For each measured subsystem $k$:\\
05   \>\>Choose $h_k$ according to the probability distribution $\func{Pr}[h_k=w] = \abs{\brakket{w}{U^s_k}{j_k}}^2$.\\
06   \>Output $\vec{h}$.
\end{tabbing}
}
\caption{Pseudocode for simulating a conventional concordant computation.  $U^0$ and $U^s$ are unitary product operators identifying the initial and final product eigenbases respectively and $P$ is the permutation that acts on $\rho^0$ equivalently to the specified sequence of unitary operators.  Pseudocode for converting a sequence of two-qubit unitary operators in a concordant computation into an equivalent classical permutation and change of basis is given in Fig.~\ref{fig:circuitConversionPseudocode}.\label{fig:concordantComputationSimulationPseudocode}}
\end{figure*}

\section{Simulating a conventional concordant computation \label{sec:simulatingConventionalConcordantComputations}}

In the previous section I show that the transformation of one concordant state to another by a unitary operator with support $\set{G}$ is equivalent to a permutation of eigenvalues together with a change of product eigenbasis on $\set{G}$.  Combined with the fact that a density operator can be considered as a probabilistic mixture of its eigenstates, this suggests the following strategy for simulating a conventional concordant computation:  Find a change of product eigenbasis and permutation of the vectors labeling eigenstates (and, therefore, the associated eigenvalues) equivalent to each unitary gate in the computation, and then generate an output of the computation by appropriately picking a vector labeling an eigenstate of the input state, applying the derived permutations to the chosen vector, and evaluating the final measurement on the indicated product state.

It is not immediately obvious that the described simulation is feasible because the permutation and change of eigenbasis equivalent to each unitary operator is dependent on the overall state of the computer.
Nonetheless, the following subsections provide detailed descriptions of the necessary subcomponents of such a simulation for the special case of two-qubit unitary gates, thereby proving Theorem~\ref{thm:twoQubitConcordantCompsSimulable}.
Section~\ref{subsec:simulationGivenHints} shows how a conventional concordant computation can be simulated given the permutation and eigenbasis change equivalent to each unitary operator.  Section~\ref{subsec:updatingProductBasis} proves that it is possible to efficiently determine a permutation and change of eigenbasis equivalent to a unitary operator from the degeneracy of the pre-gate state.  Finally, Sec.~\ref{subsec:diagnoseDegeneracy} explains how the relevant degeneracy can be found from the previously applied permutations and an input product state, so long as the computation contains only one- and two-qubit unitary gates.
In addition to the concordant-state condition given by Eq.~\ref{eq:concordantStateCriterion2}, I employ an equivalent definition: a state $\rho$ is concordant if and only if there exists a unitary product operator $U=\bigotimes_k U_k$ such that $U^\dagger \rho U$ is diagonal in the standard basis.

\begin{thm}
A conventional concordant computation with unitary operators having support on only one or two qubits can be efficiently simulated by a classical computer.\label{thm:twoQubitConcordantCompsSimulable}
\end{thm}

\begin{figure*}
\code{
\begin{tabbing}
01\hspace{2ex}\=Stor\=e th\=e un\=itar\=y op\=erat\=or d\=efining the initial product eigenbasis in $U$.\\
02   \>$P := I$ \>\>\>\>\>\>\>\textcolor{gray}{(where $P$ is stored as a sequence of two-bit permutations)}\\
03   \>For each gate $G$ in the circuit:\\
04   \>\>If $G$ has support on only one qubit:\\
05   \>\>\>$U:=G U$\\
06   \>\>Else if $G$ has support on some pair of qubits $\set{G}=\{k,l\}$:\\
07   \>\>\>For each permutation $Q$ which exchanges two states of the standard basis of part $\set{G}$:\\
08   \>\>\>\>If $P^\dagger Q P$ commutes with the initial density operator:\\
09   \>\>\>\>\>The states exchanged by $Q$ are degenerate.  Store this fact.\\
10   \>\>\>Solve for $V$, and thus the new product eigenbasis, using the known degeneracy and the constraint\\
\>\>\>that the post-gate state be diagonal in that basis.\\
11   \>\>\>Pick a permutation $R$ such that $V R U^\dagger$ and $G$ transform the state identically.\\
12   \>\>\>$P := R P$ \\
13   \>\>\>$U := V$\\
14   \>Output $P$ and $U$.
\end{tabbing}
}
\caption{Pseudocode for converting the sequence of unitary gates in a conventional concordant computation composed of one- and two-qubit gates to an equivalent permutation and change of basis.\label{fig:circuitConversionPseudocode}}
\end{figure*}

\subsection{Simulation given many hints \label{subsec:simulationGivenHints}}

Consider a conventional concordant computation for which the sequence of unitary operators employed, $\{G^t\}$, is known to act equivalently to the sequence $\left\{U^t P^t {U^{t-1}}^\dagger\right\}$ where each $P^t$ is a permutation (that is, a classical reversible gate) with the same support as $G^t$ and each $U^t$ is a unitary product operator that transforms from the standard basis to the product eigenbasis at time step $t$.
Given this information, the initial state $\rho^0$ must be of the form
\begin{align}
  \rho^0 &= \sum_{\vec{i}} p^0_{\vec{i}} U^0\ket{\vec{i}}\bra{\vec{i}}{U^0}^\dagger
\end{align}
where each $\ket{\vec{i}}$ is an element of the standard basis. (By definition, $U^0$ is trivial for a conventional computation.)
The state of the computer after one step of the computation is
\begin{align*}
  \begin{split}
    \rho^1 &= \sum_{\vec{i}} p^0_{\vec{i}} G^1 U^0\ket{\vec{i}}\bra{\vec{i}}{U^0}^\dagger {G^1}^{\dagger} \\
    &= \sum_{\vec{i}} p^0_{\vec{i}} U^1 {U^1}^\dagger G^1 U^0\ket{\vec{i}}\bra{\vec{i}}{U^0}^\dagger {G^1}^{\dagger} U^1 {U^1}^\dagger\\
    &= \sum_{\vec{i}} p^0_{\vec{i}} U^1 P^1 \ket{\vec{i}}\bra{\vec{i}}{P^1}^{\dagger} {U^1}^\dagger
  \end{split}
\end{align*}
where $P^1$ is a permutation that acts identically to ${U^1}^\dagger G^1 U^0$ on ${U^0}^\dagger \rho^0 U^0$.  Iterating this process yields
\begin{align}
    \rho^s &= \sum_{\vec{i}} p^0_{\vec{i}} U^s \left(\prod_{t=s}^1 P^t\right) \ket{\vec{i}}\bra{\vec{i}} \left(\prod_{t=1}^s {P^t}^{\dagger}\right) {U^s}^\dagger \label{eq:preMeasurementState}
\end{align}
where each $P^t$ is a permutation that acts identically to ${U^t}^\dagger G^t U^{t-1}$ on ${U^{t-1}}^\dagger \rho^{t-1} U^{t-1}$.

The measurement statistics of a mixed state are identical to those of a probabilistically chosen state in its decomposition where the probability is given by the coefficient of the term associated with that state.
Thus, the expression for the final pre-measurement state shown in Eq.~\ref{eq:preMeasurementState} suggests the following simple technique for simulating the computation:  Choose a single vector $\vec{i}$ according to the probability distribution $p^0_{\vec{i}}$, which can be done efficiently since $\rho^0$ is a product state.  Apply the permutation $\prod_{t=s}^1 P^t$ to $\vec{i}$ to obtain a new vector $\vec{j}$ identifying one component of the final pre-measurement state.  And last, for each measured subsystem $k$ choose a measurement outcome $h_k$ according to the probability distribution
\begin{align*}
  \func{Pr}[h_k=w] = \abs{\brakket{w}{{U^s_k}}{j_k}}^2\;.
\end{align*}
Fig.~\ref{fig:concordantComputationSimulationPseudocode} presents pseudocode illustrating this method.

\subsection{Updating the product eigenbasis\label{subsec:updatingProductBasis}}

In the $t$th step of a conventional concordant computation, the unitary gate $G^t$ is applied to a concordant state $\rho^{t-1}$ to yield a concordant state $\rho^t$.  As explained in Sec.~\ref{subsec:concordantComputations}, the effect of $G^t$ is identical to that of a permutation $P^t$ of the vectors labeling eigenstates followed by a change of product eigenbasis.  Thus, if $U^{t-1}$ and $U^t$ are unitary product operators that transform from the standard basis to the product eigenbases at times $t-1$ and $t$, respectively, then $\rho^t=G^t\rho^{t-1}{G^t}^\dagger = U^t P^t {U^{t-1}}^\dagger \rho^{t-1} U^{t-1} {P^t}^\dagger {U^t}^\dagger$ for some $P^t$ which permutes the elements of the standard basis.
Moreover, Lemma~\ref{lem:concordantComputationsAreSimple} shows that there exists a product eigenbasis for $\rho^t$ consistent with $U^t$ such that $U^t_k = U^{t-1}_k$ for all $k$ not in $\set{G}^t$, the support of $G^t$, and additionally, that for such a product eigenbasis there exists a permutation $P^t$ with support $\set{G}^t$.

The problem of finding $U^t_k$ for $k\not\in\set{G}^t$ is addressed by Lemma~\ref{lem:eigenbasisFromPartDegeneracy}, which shows that the remaining components of a product eigenbasis for $\rho^t$ can be calculated given one additional piece of information, the degeneracy of part $\set{G}^t$ of $\rho^{t-1}$.
This calculation is efficient in that it entails solving a system of equations whose number depends only on the number of subsystems in $\set{G}^t$ and their dimension, not on the total number of subsystems in the computation.
The appropriate permutation is easily found from the eigenbases for $\rho^{t-1}$ and $\rho^t$; it is sufficient to pick any permutation mapping eigenprojectors of $\rho^{t-1}$ to eigenprojectors of $\rho^t$ which are in the same connected component of a graph $\Gamma$ defined as per Lemma~\ref{lem:concordantComputationsAreSimple}.
(Remember that the permutation $P^t$ can be assumed to have support $\set{G}^t$, thereby limiting the size of the graph that must be considered.)
As indicated by Theorem~\ref{thm:simulationForGenericStates}, these results are sufficient to enable the efficient simulation of concordant computations with most input states.  The question of arbitrary input states is taken up in the next section.

\begin{lem}
For $\rho = \sum_{\vec{i}} p_{\vec{i}} \proj{P}^{\vec{i}}$ and $\sigma = G \rho G^\dagger$, where $G$ is a unitary gate with support $\set{G}$,
$\{\proj{Q}^{\vec{j}}\}$ satisfies $\sigma  = \sum_{\vec{j}} \proj{Q}^{\vec{j}}_\set{G} \sigma \proj{Q}^{\vec{j}}_\set{G}$ if and only if $\tr_\set{G}(\rho \proj{P}^{\vec{i}}_\set{G})=\tr_\set{G}(\rho \proj{P}^{\vec{h}}_\set{G})$ for all $\vec{h}$, $\vec{i}$, and $\vec{j}$ such that $G \proj{P}^{\vec{h}}_\set{G} G^\dagger \proj{Q}^{\vec{j}}_\set{G} \ne 0$ and $G \proj{P}^{\vec{i}}_\set{G} G^\dagger \proj{Q}^{\vec{j}}_\set{G} \ne 0$. \label{lem:eigenbasisFromPartDegeneracy}
\end{lem}

\begin{proof}
\begin{align*}
\sum_{\vec{i}} G \proj{P}^{\vec{i}}_\set{G} G^\dagger \sigma G \proj{P}^{\vec{i}}_\set{G} G^\dagger = \sum_{\vec{i}} G \proj{P}^{\vec{i}}_\set{G} \rho \proj{P}^{\vec{i}}_\set{G} G^\dagger = G \rho G^\dagger = \sigma\;,
\end{align*}
so by Lemma~\ref{lem:projectorsUniqueUpToPartialDegeneracy}, $\{\proj{Q}^{\vec{j}}\}$ satisfies $\sigma  = \sum_{\vec{j}} \proj{Q}^{\vec{j}}_\set{G} \sigma \proj{Q}^{\vec{j}}_\set{G}$ if and only if $\tr_\set{G}(\sigma G \proj{P}^{\vec{h}}_\set{G} G^\dagger)=\tr_\set{G}(\sigma \proj{Q}^{\vec{j}}_\set{G})$ for all $G \proj{P}^{\vec{h}}_\set{G} G^\dagger \proj{Q}^{\vec{j}}_\set{G} \ne 0$.
Given $\rho$, $\sigma$, $G$, $\{\proj{P}^{\vec{i}}\}$, and $\{\proj{Q}^{\vec{j}}\}$ as defined, the condition $\tr_\set{G}(\sigma G \proj{P}^{\vec{h}}_\set{G} G^\dagger)=\tr_\set{G}(\sigma \proj{Q}^{\vec{j}}_\set{G})$ for all $G \proj{P}^{\vec{h}}_\set{G} G^\dagger \proj{Q}^{\vec{j}}_\set{G} \ne 0$ is equivalent to $\tr_\set{G}(\rho \proj{P}^{\vec{h}}_\set{G})=\tr_\set{G}(\rho \proj{P}^{\vec{i}}_\set{G})$ for all $G \proj{P}^{\vec{h}}_\set{G} G^\dagger \proj{Q}^{\vec{j}}_\set{G} \ne 0$ and $G \proj{P}^{\vec{i}}_\set{G} G^\dagger \proj{Q}^{\vec{j}}_\set{G} \ne 0$.
\end{proof}

Fig.~\ref{fig:circuitConversionPseudocode} presents pseudocode for an algorithm calculating the necessary sequence of permutations and basis changes.

\begin{thm}
A conventional concordant computation with an input product state that is generic can be efficiently simulated by a classical computer.\label{thm:simulationForGenericStates}
\end{thm}
\begin{proof}
A generic product state has no degenerate eigenvalues, so the simulation method as outlined thus far is sufficient for such input states.
\end{proof}

\subsection{Diagnosing the degeneracy\label{subsec:diagnoseDegeneracy}}

In order to update the product eigenbasis following the $t$th gate in a conventional concordant computation, it is necessary to diagnose the degeneracy of part $\set{G}^t$ of $\rho^{t-1}$, where $\set{G}^t$ is the support of $G^t$, the $t$th gate in the computation, and $\rho^{t-1}$ is the state of the computation at time $t-1$.  This degeneracy can be found by determining whether $\rho^{t-1}$ and $U^{t-1} Q {U^{t-1}}^\dagger$ commute for each permutation $Q$ exchanging two eigenstates of the standard basis for the subsystems in $\set{G}^{t}$.  As the simulation algorithm progresses, permutations equivalent to each gate are found, so $\rho^{t-1} = U^{t-1} P \rho^0 P^\dagger {U^{t-1}}^\dagger$ where $P = \prod_r^{t-1} P^r$ represents the sequence of (known) permutations up to step $t-1$.  Thus, one may equally well check whether
\begin{align}
  \rho^0 = P^\dagger Q P \rho^0 P^\dagger Q P \label{eq:commutativityConstraint1}\;.
\end{align}

I now restrict my attention to concordant computations composed of two-qubit gates acting on a register of $n$ qubits.
The permutation $P^\dagger Q P$ is an involution, i.e., it is self-inverse, and for the case of qubits and two-qubit gates, it is affine when considered as a function on binary vectors.
Lemma~\ref{lem:testingCommutativity} shows that such a permutation commutes with $\rho^0$ if and only if Eq.~\ref{eq:commutativityConstraint1} is satisfied for the pure product state corresponding to each of a particular set of $n+1$ binary vectors.  Consequently, the commutativity of $\rho^0$ and $P^\dagger Q P$, and therefore the degeneracy relevant to updating the product eigenbasis, can be efficiently determined for concordant computations composed of two-qubit gates.

\begin{lem}
A product state on qubits, $\rho = \bigotimes_k \rho_k$, such that $\rho$ is diagonal in the standard basis and $e_k=\brakket{1}{\rho_k}{1}/\brakket{0}{\rho_k}{0}\leq 1$ for all $k$ commutes with an affine involution $S$ if and only if
$\brakket{\vec{i}}{S\rho S^\dagger - \rho}{\vec{i}}=0$ for all $\ket{\vec{i}}$ such that $i_k=\updelta_{k l}$ or $i_k=0$.
\label{lem:testingCommutativity}
\end{lem}

\begin{proof}\

Throughout this proof, binary vectors labeling states are represented by the set of indices identifying bits in the $\ket{1}$ state.  Let $\set{S}$ be a version of $S$ that acts on such sets\footnote{For brevity I omit brackets in the argument of this and other functions when the input is a singleton, e.g., I write $\set{S}(k)$ rather than $\set{S}(\{k\})$.}.  In this representation, the affine linearity of $\set{S}$ is expressed as $\set{S}(\set{A}\ominus\set{B}) = \set{S}(\set{A})\ominus\set{S}(\set{B})\ominus\set{K}$ for some fixed $\set{K}$, while the fact that $\set{S}$ is an involution implies that $\set{S}(\set{S}(\set{A}))=\set{A}$.  Define $\set{C}_e = \{k: e_k = e\}$ and $f(\set{B})=\prod_{k\in \set{B}} e_k$.  In terms of $f$ and $\set{S}$ the commutativity condition to be satisfied is
\begin{align}
  f(\set{B}) = f(\set{S}(\set{B})) \label{eq:commutativityConstraint2}
\end{align}
for any set of bits, $\set{B}$.

The forward implication stated in this lemma is trivial.  If Eq.~\ref{eq:commutativityConstraint2} is satisfied for any set $\set{B}$ then it is obviously satisfied for any singleton $\{k\}$ and for the empty set.

To demonstrate the reverse, I assume, for the remainder of the proof, that Eq.~\ref{eq:commutativityConstraint2} is satisfied for the empty set and any singleton and seek to show that it is satisfied in general.  I organize what follows in terms of a sequence of small points.\\

\noindent\textbf{Point 0:} $\set{T}(\set{B})=\set{S}(\set{B})\ominus\set{K}$ is a linear involution, and $\set{T}$ satisfies Eq.~\ref{eq:commutativityConstraint3} if and only if $\set{S}$ satisfies Eq.~\ref{eq:commutativityConstraint2}.\\
$\set{T}$ is linear since $\set{S}$ is affine with constant $\set{K}$.  Because $\set{S}(\emptyset)=\set{K}$ and $\set{S}$ is an involution, $\set{S}(\set{K})=\set{S}(\set{S}(\emptyset))=\emptyset$, implying that $\set{T}$ is an involution since
\begin{align*}
  \begin{split}
    \set{T}(\set{T}(\set{B})) &= \set{T}(\set{S}(\set{B})\ominus\set{K}) = \set{S}(\set{S}(\set{B})\ominus\set{K}) \ominus \set{K} \\
    &= \set{S}(\set{S}(\set{B}))\ominus\set{S}(\set{K})\ominus \set{K} \ominus \set{K} = \set{B} \;.
  \end{split}
\end{align*}
Furthermore, $\set{K}\subseteq\set{C}_1$ since if $\exists k\in\set{K}$ such that $k\not\in\set{C}_1$ then $f(\set{K})\leq e_k < 1 =f(\emptyset)$.  Consequently, $f(\set{B}) = f(\set{B}\ominus\set{K})$, and thus Eq.~\ref{eq:commutativityConstraint2} is satisfied if and only if
\begin{align}
  f(\set{B}) = f(\set{T}(\set{B})) \label{eq:commutativityConstraint3}
\end{align}

\noindent\textbf{Point 1:} $\forall k \exists m\in\set{T}(k)$ such that $k\in\set{T}(m)$\\
Because $\set{T}$ is a linear involution,
\begin{align*}
  k = \set{T}(\set{T}(k)) = \set{T}\left(\bigominus_{l\in\set{T}(k)} \{l\}\right) = \bigominus_{l\in\set{T}(k)} \set{T}(l)\;,
\end{align*}
so $\forall k \exists m\in\set{T}(k)$ such that $k\in\set{T}(m)$.
\\

\noindent\textbf{Point 2:} $e_l\geq e_k \ \forall l\in\set{T}(k)$\\
If $\exists l\in\set{T}(k)$ such that $e_l<e_k$ then $f(k)=e_k>e_l\geq f(\set{T}(k))$, so  $e_l\geq e_k \ \forall l\in\set{T}(k)$.\\

\noindent\textbf{Point 3:} $\exists m\in\set{T}(k)$ such that $e_m = e_k$\\
By the previous two points $\exists m\in\set{T}(k)$ such that $k\in\set{T}(m)$ and $e_m\geq e_k$, but this implies that $e_m = e_k$ since, by Point 2, $k\in\set{T}(m)$ implies $e_k \geq e_m$.\\

\noindent\textbf{Point 4:} Each $k\in\set{C}_e$ where $e>0$ is mapped by $\set{T}$ to a single $m\in\set{C}_e$ together with (possibly) some elements of $\set{C}_1$.\\
By the previous point, $\exists m\in\set{T}(k)$ such that $m\in\set{C}_{e_k}$, implying that
\begin{align*}
  f(\set{T}(k)) = f(m) f(\set{T}(k)\setminus\{m\}) = e_k \prod_{l\in\set{T}(k)\setminus\{m\}} e_l\;,
\end{align*}
which is equal to $f(k)=e_k$ only when $e_k=0$ or $e_l=1 \ \forall l \in \set{T}(k)\setminus\{m\}$. \\

\noindent\textbf{Point 5:} Any two distinct elements $k,l\in\set{C}_e$ where $e>0$ are mapped by $\set{T}$ to distinct elements of $\set{C}_e$ together with (possibly) some elements of $\set{C}_1$.\\
If $\exists k,l\in\set{C}_e$ with $k\ne l$ and $1>e>0$ such that $\set{T}(k)/\set{C}_1=\set{T}(l)/\set{C}_1= \{m\}$ then $k,l\in\set{T}(m)$ since $k,l\not\in\set{T}(o)$ for any $o\in\set{C}_1$, which contradicts the preceding point.\\

\noindent\textbf{Point 6:} If $\set{B}\cap\set{C}_0\neq \emptyset$ then $\set{T}(\set{B})\cap\set{C}_0\neq\emptyset$.\\
If $\exists \set{B}$ such that $\set{B}\cap\set{C}_0\neq\emptyset$ but $\set{T}(\set{B})\cap\set{C}_0=\emptyset$ then $\exists l\in\set{T}(\set{B})$ such that $e_l>0$ and $\set{T}(l)\cap\set{C}_0\neq\emptyset$, which contradicts my second point.\\

\noindent\textbf{Point 7:} Eq.~\ref{eq:commutativityConstraint3} is satisfied for any set $\set{B}$.\\
If $\set{B}\cap\set{C}_0\neq \emptyset$ then $\set{T}(\set{B})\cap\set{C}_0\neq\emptyset$ so $f(\set{B})=f(\set{T}(\set{B})) = 0$.  Otherwise,
\begin{align*}
  \begin{split}
    f(\set{B}) = \prod_{l\in\set{B}} f(l) &= \prod_{l\in\set{B}} f(\set{T}(l)) = f\left(\bigominus_{l\in\set{B}} \set{T}(l)\right)\\
    &= f\left(\set{T}\left(\bigominus_{l\in\set{B}} \{l\}\right)\right) = f(\set{T}(\set{B}))\;,
  \end{split}
\end{align*}
where the middle equality follows from Point 5, which shows that $\set{T}(l)\cap\set{T}(k)\subseteq\set{C}_1$ for all $k$ and $l$ such that $k\ne l$ and $k,l\not\in\set{C}_0,\set{C}_1$.
\end{proof}

\section{Extensions \label{sec:extensions}}

Quantum computations, even those described in terms of quantum circuits, frequently are not envisioned in the conventional form outlined by Def.~\ref{def:conventionalQuantumComputation}.  The most common deviations are the inclusion of single-subsystem measurements intermixed with the unitary operators and the introduction of new subsystems during the course of the computation.  Another possibility for concordant computations is that the input state be a mixture of product states that is not also a product of mixed states but that can be efficiently prepared due to the mixture having few terms.  Computations with these features can be converted to conventional ones (allowing for some post selection to assist in the generation of the desired input state), but, in general, the conversion process preserves neither the concordance of the computation nor the maximal support of its unitary operators.  While subsystems introduced during the course of a computation can equally well be introduced at its beginning, non-terminal measurements and non-product-state inputs require special treatment.

\subsection{Non-terminal measurements}

It requires some effort to extend the simulation algorithm described in the previous section to non-terminal measurements on single subsystems.  Through the first measurement, the simulation may proceed exactly as previously explained, but subsequent to that, a more complex technique for diagnosing the degeneracy is necessary since measurements introduce the possibility that the degeneracy relevant to determining the permutation and change of eigenbasis equivalent to a gate might be dependent on the outcome of the measurement result.  There seems to be a method of efficiently diagnosing the relevant degeneracy when measurements are performed in the eigenbasis, but the more general problem is one that I have not yet been able to solve.

\subsection{Non-product-state inputs}

Generically, Def.~\ref{def:conventionalQuantumComputation} excludes a very natural kind of mixed input state, namely, the probabilistic mixture of a few pure product states.  As it happens, however, concordant computations with such input states are easy to simulate; the state of the computer can simply be stored and updated explicitly.  The algorithm is the same as that described in Sec.~\ref{sec:simulatingConventionalConcordantComputations} except that the degeneracy is straightforward to evaluate since the state is explicitly known.  Because unitary operators do not change the rank of density matrix and projective measurements can only decrease it, explicit storage of the state remains practical throughout the simulation.

Effectively, a quantum computation on a low-rank input state becomes complicated only because the eigenbasis becomes complicated.  For a concordant computation the eigenbasis remains manageable.

\section{Conclusion}

In summary, I have shown that conventional concordant computations composed exclusively of gates acting on one or two qubits can be efficiently simulated using a classical computer.  As a consequence, such a computation must generate quantum discord if it is to permit the efficient solution of a problem requiring super-polynomial resources classically.  A similar statement holds for more general gate sets whenever the input state is either a generic product state or a mixture of a few pure product states.  These results lend support to the idea that quantum discord is the appropriate generalization of entanglement with regard to mixed-state quantum computation.  That being said, concordance is such a stringent property that it no doubt corresponds to the case of zero quantum correlations for a variety of measures (including the many flavors of discord), so this is far from the final word on the subject.  As has periodically been noted, it is also important to keep in mind that there can be no single resource for quantum computing: If quantum computations without property $\prop{P}$ can be efficiently simulated classically then $\prop{P}$ is a necessary resource for achieving a Promethean speedup.

Several possible directions for future research are suggested by previous work on simulating quantum computations with restricted entanglement.  The two most prominent are investigating the performance of the simulation for approximately concordant states and extending it to computations where discord is restricted to blocks of qubits of bounded size.  A block of qubits with unrestricted correlations can be treated as a single quantum system, so progress on the latter topic would likely require extending the simulation method to qudits.

Though I specialize to qubits and two-qubit gates only in Sec.~\ref{subsec:diagnoseDegeneracy}, it is doubtful whether my simulation method can be extended to more general gate sets.
Section~\ref{subsec:diagnoseDegeneracy} depends crucially on the fact that permutations on one or two bits of a vector are necessarily linear (or, from an alternate perspective, that such permutations are Clifford gates) since this allows me to determine whether Eq.~\ref{eq:commutativityConstraint1} is satisfied by checking a small set of basis vectors.  On the other hand, permutations on systems of dimension greater than two or on more than two bits need not be linear.  Thus, directly generalizing the method of simulation described in this paper requires a means of testing Eq.~\ref{eq:commutativityConstraint1} for an arbitrary sequence of permutations and input (mixed) product state.  This implies the ability to efficiently solve 3-SAT, an NP-Complete problem, since $P$ in Eq.~\ref{eq:commutativityConstraint1} can be chosen to implement a boolean formula, $Q$ to copy the result to an ancillary qubit, and $\rho^0$ to consist of unbiased input qubits and maximally biased ancillary qubits, yielding $\rho^0\neq P^\dagger Q P \rho^0 P^\dagger Q P$ if and only if the boolean formula is satisfied for some input.  In other words, a direct extension of my simulation method is effectively ruled out, though I am unable to exclude the possibility that some more generally applicable method exists for simulating concordant computations.

\acknowledgments

I am grateful to Emanuel Knill, Anil Shaji, Carlton Caves, Vaibhav Madhok, and Adam Meier for many productive discussions.  This paper is a contribution by the National Institute of Standards and Technology and, as such, is not subject to U.S. copyright.

\bibliography{../../citations}

\begin{thebibliography}{12}
\expandafter\ifx\csname natexlab\endcsname\relax\def\natexlab#1{#1}\fi
\expandafter\ifx\csname bibnamefont\endcsname\relax
  \def\bibnamefont#1{#1}\fi
\expandafter\ifx\csname bibfnamefont\endcsname\relax
  \def\bibfnamefont#1{#1}\fi
\expandafter\ifx\csname citenamefont\endcsname\relax
  \def\citenamefont#1{#1}\fi
\expandafter\ifx\csname url\endcsname\relax
  \def\url#1{\texttt{#1}}\fi
\expandafter\ifx\csname urlprefix\endcsname\relax\def\urlprefix{URL }\fi
\providecommand{\bibinfo}[2]{#2}
\providecommand{\eprint}[2][]{\url{#2}}

\bibitem[{\citenamefont{Horodecki et~al.}(2009)\citenamefont{Horodecki,
  Horodecki, Horodecki, and Horodecki}}]{Horodecki09}
\bibinfo{author}{\bibfnamefont{R.}~\bibnamefont{Horodecki}},
  \bibinfo{author}{\bibfnamefont{P.}~\bibnamefont{Horodecki}},
  \bibinfo{author}{\bibfnamefont{M.}~\bibnamefont{Horodecki}},
  \bibnamefont{and}
  \bibinfo{author}{\bibfnamefont{K.}~\bibnamefont{Horodecki}},
  \bibinfo{journal}{Rev. Mod. Phys.} \textbf{\bibinfo{volume}{81}},
  \bibinfo{eid}{865} (\bibinfo{year}{2009}), \eprint{arXiv:quant-ph/0702225}.

\bibitem[{\citenamefont{Jozsa and Linden}(2003)}]{Jozsa03}
\bibinfo{author}{\bibfnamefont{R.}~\bibnamefont{Jozsa}} \bibnamefont{and}
  \bibinfo{author}{\bibfnamefont{N.}~\bibnamefont{Linden}},
  \bibinfo{journal}{Proc. R. Soc. A} \textbf{\bibinfo{volume}{459}},
  \bibinfo{pages}{2011} (\bibinfo{year}{2003}),
  \eprint{arXiv:quant-ph/0201143}.

\bibitem[{\citenamefont{Vidal}(2003)}]{Vidal03}
\bibinfo{author}{\bibfnamefont{G.}~\bibnamefont{Vidal}},
  \bibinfo{journal}{Phys. Rev. Lett.} \textbf{\bibinfo{volume}{91}},
  \bibinfo{pages}{147902} (\bibinfo{year}{2003}),
  \eprint{arXiv:quant-ph/0301063}.

\bibitem[{\citenamefont{Or\'us and Latorre}(2004)}]{Orus04}
\bibinfo{author}{\bibfnamefont{R.}~\bibnamefont{Or\'us}} \bibnamefont{and}
  \bibinfo{author}{\bibfnamefont{J.~I.} \bibnamefont{Latorre}},
  \bibinfo{journal}{Phys. Rev. A} \textbf{\bibinfo{volume}{69}},
  \bibinfo{pages}{052308} (\bibinfo{year}{2004}),
  \eprint{arXiv:quant-ph/0311017}.

\bibitem[{\citenamefont{Knill and Laflamme}(1998)}]{Knill98b}
\bibinfo{author}{\bibfnamefont{E.}~\bibnamefont{Knill}} \bibnamefont{and}
  \bibinfo{author}{\bibfnamefont{R.}~\bibnamefont{Laflamme}},
  \bibinfo{journal}{Phys. Rev. Lett.} \textbf{\bibinfo{volume}{81}},
  \bibinfo{pages}{5672} (\bibinfo{year}{1998}),
  \eprint{arXiv:quant-ph/9802037}.

\bibitem[{\citenamefont{A.~Ambainis and Vazirani}(2006)}]{Ambainis06}
\bibinfo{author}{\bibfnamefont{L.~S.} \bibnamefont{A.~Ambainis}}
  \bibnamefont{and} \bibinfo{author}{\bibfnamefont{U.}~\bibnamefont{Vazirani}},
  \bibinfo{journal}{Journal of the ACM} \textbf{\bibinfo{volume}{53}},
  \bibinfo{pages}{507} (\bibinfo{year}{2006}), \eprint{arXiv:quant-ph/0003136}.

\bibitem[{\citenamefont{Datta et~al.}(2005)\citenamefont{Datta, Flammia, and
  Caves}}]{Datta05}
\bibinfo{author}{\bibfnamefont{A.}~\bibnamefont{Datta}},
  \bibinfo{author}{\bibfnamefont{S.~T.} \bibnamefont{Flammia}},
  \bibnamefont{and} \bibinfo{author}{\bibfnamefont{C.~M.} \bibnamefont{Caves}},
  \bibinfo{journal}{Phys. Rev. A} \textbf{\bibinfo{volume}{72}},
  \bibinfo{pages}{042316} (\bibinfo{year}{2005}),
  \eprint{arXiv:quant-ph/0505213}.

\bibitem[{\citenamefont{Datta and Vidal}(2007)}]{Datta07}
\bibinfo{author}{\bibfnamefont{A.}~\bibnamefont{Datta}} \bibnamefont{and}
  \bibinfo{author}{\bibfnamefont{G.}~\bibnamefont{Vidal}},
  \bibinfo{journal}{Phys. Rev. A} \textbf{\bibinfo{volume}{75}},
  \bibinfo{pages}{042310} (\bibinfo{year}{2007}),
  \eprint{arXiv:quant-ph/0611157}.

\bibitem[{\citenamefont{Zurek}(2000)}]{Zurek00}
\bibinfo{author}{\bibfnamefont{W.~H.} \bibnamefont{Zurek}},
  \bibinfo{journal}{Ann. Phys.} \textbf{\bibinfo{volume}{9}},
  \bibinfo{pages}{855} (\bibinfo{year}{2000}), \eprint{arXiv:quant-ph/0011039}.

\bibitem[{\citenamefont{Datta et~al.}(2008)\citenamefont{Datta, Shaji, and
  Caves}}]{Datta08}
\bibinfo{author}{\bibfnamefont{A.}~\bibnamefont{Datta}},
  \bibinfo{author}{\bibfnamefont{A.}~\bibnamefont{Shaji}}, \bibnamefont{and}
  \bibinfo{author}{\bibfnamefont{C.~M.} \bibnamefont{Caves}},
  \bibinfo{journal}{Physical Review Letters} \textbf{\bibinfo{volume}{100}},
  \bibinfo{eid}{050502} (\bibinfo{year}{2008}), \eprint{arXiv:0709.0548}.

\bibitem[{\citenamefont{Lanyon et~al.}(2008)\citenamefont{Lanyon, Barbieri,
  Almeida, and White}}]{Lanyon08}
\bibinfo{author}{\bibfnamefont{B.~P.} \bibnamefont{Lanyon}},
  \bibinfo{author}{\bibfnamefont{M.}~\bibnamefont{Barbieri}},
  \bibinfo{author}{\bibfnamefont{M.~P.} \bibnamefont{Almeida}},
  \bibnamefont{and} \bibinfo{author}{\bibfnamefont{A.~G.} \bibnamefont{White}},
  \bibinfo{journal}{Phys. Rev. Lett.} \textbf{\bibinfo{volume}{101}},
  \bibinfo{eid}{200501} (\bibinfo{year}{2008}), \eprint{arXiv:0807.0668}.

\bibitem[{\citenamefont{Ollivier and Zurek}(2001)}]{Ollivier01}
\bibinfo{author}{\bibfnamefont{H.}~\bibnamefont{Ollivier}} \bibnamefont{and}
  \bibinfo{author}{\bibfnamefont{W.~H.} \bibnamefont{Zurek}},
  \bibinfo{journal}{Phys. Rev. Lett.} \textbf{\bibinfo{volume}{88}},
  \bibinfo{pages}{017901} (\bibinfo{year}{2001}),
  \eprint{arXiv:quant-ph/0105072}.

\end{thebibliography}

\end{document}